\documentclass{llncs}
  \usepackage{amssymb,amsmath}
  \usepackage{verbatim}
  \usepackage{enumerate}
  \usepackage{graphicx}
  \usepackage{tikz}
  \usetikzlibrary{snakes}
  \usepackage[noend]{algorithmic}
  \usepackage[boxed]{algorithm}

  \renewcommand{\algorithmicrequire}{\textbf{ALGORITHM}}
  \floatname{algorithm}{Algorithm}

  \newtheorem{fact}[lemma]{Fact}
  \newtheorem{observation}[lemma]{Observation}

  \newcommand{\per}{\textsf{per}}
  \newcommand{\border}{\textsf{border}}

  \newcommand{\SUF}{\textsf{SUF}}
  \newcommand{\LCP}{\textsf{LCP}}
  
  \newcommand{\Q}{\textsl{FirstGE}}

  \newcommand{\cover}{\textsf{cover}}
  \newcommand{\covermax}{\textsf{covermax}}
  \newcommand{\lseed}{\textsf{lseed}}
  \newcommand{\lseedmax}{\textsf{lseedmax}}
  \newcommand{\rseed}{\textsf{rseed}}
  \newcommand{\bseed}{\textsf{bseed}}
  \newcommand{\seed}{\textsf{seed}}

  \newcommand{\Cover}{\textsf{C}}
  \newcommand{\LSeed}{\textsf{LSeed}}
  \newcommand{\Seed}{\textsf{Seed}}

  \newcommand{\Border}{\textsf{B}}
  \newcommand{\Period}{\textsf{P}}
  \newcommand{\R}{\textsf{R}}

  \newcommand{\Tree}{T}
  \newcommand{\Nodes}{\mbox{\textit{Nodes}}}
  \newcommand{\firstocc}{\mbox{\textit{first}}}
  \newcommand{\lastocc}{\mbox{\textit{last}}}
  \newcommand{\Occ}{\mbox{\textit{Occ}}}
  \newcommand{\LL}{\mbox{\textit{LL}}}

  \newcommand{\FF}{\mathcal{F}}

  \newcommand{\maxgap}{\textsf{maxgap}}
  \newcommand{\prefmaxgap}{\Delta}

  \newcommand{\SORT}{\mbox{\textit{SORT}}}
  \newcommand{\PREF}{\mbox{\textit{PREF}}}

  \def\dotdot{\mathinner{\ldotp\ldotp}}

  \date{}
  \author{\bf
    Michalis Christou\inst{1}
    \and
    Maxime Crochemore\inst{1}\fnmsep\inst{3}
    \and
    Costas S.\ Iliopoulos\inst{1}\fnmsep\inst{4}
    \and \\
    Marcin Kubica\inst{2}
    \and
    Solon P.\ Pissis\inst{1}
    \and
    Jakub Radoszewski\inst{2}\thanks{
      The author is supported by grant no.\ N206 568540 of the National Science Centre.
    }
    \and \\
    Wojciech Rytter
    \inst{2}\fnmsep\inst{5}\thanks{
    The author is supported by grant no.\ N206 566740 of the National Science Centre.
    }
    \and
    Bartosz Szreder\inst{2}
    \and
    Tomasz Wale\'n\inst{2}
  }

  \institute{
    Dept.~of Informatics, King's College London, London WC2R 2LS, UK \\
    \email{[michalis.christou,maxime.crochemore,csi,solon.pissis]@dcs.kcl.ac.uk}
    \and
    Dept.~of Mathematics, Computer Science and Mechanics, \\
    University of Warsaw, Warsaw, Poland\\
    \email{[kubica,jrad,rytter,szreder,walen]@mimuw.edu.pl}
    \and
    Universit\'e Paris-Est, France
    \and
    Digital Ecosystems \& Business Intelligence Institute, \\
    Curtin University of Technology, Perth WA 6845, Australia
    \and
    Dept. of Math. and Informatics,\\
    Copernicus University, Toru\'n, Poland
  }

  \title{
    Efficient Seeds Computation Revisited\thanks{
      The authors thank an anonymous referee for proposing several insightful remarks.
    }
  }

\begin{document}

  \maketitle
  \begin{abstract}
    The notion of the cover is a generalization of a period of a string, and
    there are linear time algorithms for finding the shortest cover.
    The seed is a more complicated generalization of periodicity, it is a cover of a superstring of
    a given string, and the shortest seed problem is of much higher algorithmic difficulty.
    The problem is not well understood, no linear time algorithm is known.
    In the paper we give linear time algorithms for some of its versions ---
    computing shortest left-seed array, longest left-seed array and checking for seeds of a given length.
    The algorithm for the last problem is used to compute the seed array of a string
    (i.e., the shortest seeds for all the prefixes of the string) in $O(n^2)$ time.
    We describe also a simpler alternative algorithm computing efficiently the shortest seeds.
    As a by-product we obtain an $O(n\log{(n/m)})$ time algorithm checking if the shortest seed has length at
    least $m$ and finding the corresponding seed.
    We also correct some important details missing in
    the previously known shortest-seed algorithm (Iliopoulos et al., 1996).
  \end{abstract}

  \section{Introduction}
    The notion of periodicity in strings is widely used in many fields, such as
    combinatorics on words, pattern matching, data compression and automata theory (see~\cite{Lot01,Lot05}).
    It is of paramount importance in several applications, not to talk about its theoretical aspects.
    The concept of quasiperiodicity is a generalization of the notion of periodicity, and
    was defined by Apostolico and Ehrenfeucht in~\cite{Apo93}.
    In a periodic repetition the occurrences of the period do not overlap.
    In contrast, the quasiperiods of a quasiperiodic string may overlap.

    We consider \emph{words} (\emph{strings}) over a finite alphabet $\Sigma$, $u \in \Sigma^*$;
    the empty word is denoted by $\varepsilon$; the positions in $u$ are numbered from $1$ to $|u|$.
    By $\Sigma^n$ we denote the set of words of length $n$.
    By $u^R$ we denote the reverse of the string $u$.
    For $u=u_1u_2\ldots u_n$, let us denote by $u[i \dotdot j]$ a \emph{factor}
    of $u$ equal to $u_i\ldots u_j$ (in particular $u[i]=u[i \dotdot i]$).
    Words $u[1 \dotdot i]$ are called \emph{prefixes} of $u$, and words $u[i \dotdot n]$ are called \emph{suffixes} of $u$.
    Words that are both prefixes and suffixes of $u$ are called \emph{borders} of $u$.
    By $\border(u)$ we denote the length of the longest border of $u$ that is shorter than $u$.
    We say that a positive integer $p$ is the (shortest) \emph{period} of a word $u=u_1\ldots u_n$
    (notation: $p=\per(u)$) if $p$ is the smallest positive number, such that $u_i=u_{i+p}$, for $i=1,\dots,n-p$.
    It is a known fact \cite{AlgorithmsOnStrings,Jewels} that, for any string $u$,
    $\per(u)+\border(u) = |u|$.

    We say that a string $s$ \emph{covers} the string $u$ if every letter of $u$ is contained
    in some occurrence of $s$ as a factor of $u$.
    Then $s$ is called a \emph{cover} of $u$.
    We say that a string $s$ is:
        a \emph{seed} of $u$ if $s$ is a factor of $u$ and $u$ is a factor of some string $w$ covered by $s$;
        a \emph{left seed} of $u$ if $s$ is both a prefix and a seed of $u$;
        a \emph{right seed} of $u$ if $s$ is both a suffix and a seed of $u$
        (equivalently, $s^R$ is a left seed of $u^R$).
    Seeds were first defined and studied by Iliopoulos, Moore and Park~\cite{DBLP:journals/algorithmica/IliopoulosMP96},
    who gave an $O(n\log{n})$ time algorithm computing all the seeds of a given string $u \in \Sigma^n$,
    in particular, the shortest seed of $u$.

    By $\cover(u)$, $\seed(u)$, $\lseed(u)$ and $\rseed(u)$ we denote the length of the shortest:
    cover, seed, left seed and right seed of $u$, respectively.
    By $\covermax(u)$ and $\lseedmax(u)$ we denote the length of the longest cover and the longest left seed of $u$
    that is shorter than $u$, or 0 if none.

    For a string $u \in \Sigma^n$, we define its:
    \emph{period array}            $\Period[1 \dotdot n]$,
    \emph{border array}            $\Border[1 \dotdot n]$,
    \emph{suffix period array}     $\Period'[1 \dotdot n]$,
    \emph{cover array}             $\Cover [1 \dotdot n]$,
    \emph{longest cover array}     $\Cover^M[1 \dotdot n]$,
    \emph{seed array}              $\Seed  [1 \dotdot n]$,
    \emph{left-seed array}         $\LSeed [1 \dotdot n]$, and
    \emph{longest left-seed array} $\LSeed^M[1 \dotdot n]$ as follows:

    \smallskip
    $
      \begin{array}{rcl@{\quad}rcl}
        \Period[i]  &=& \per(u[1 \dotdot i]),      &
        \Border[i]  &=& \border(u[1 \dotdot i]),   \\
        \Period'[i] &=& \per(u[i \dotdot n]),   &
        \Cover[i]   &=& \cover(u[1 \dotdot i]),    \\
        \Cover^M[i] &=& \covermax(u[1 \dotdot i]), &
        \Seed[i]    &=& \seed (u[1 \dotdot i]),    \\
        \LSeed[i]   &=& \lseed(u[1 \dotdot i]),    &
        \LSeed^M[i] &=& \lseedmax(u[1 \dotdot i]).
      \end{array}
    $

    \begin{table}
    $
      \begin{array}{r@{\quad}*{16}{c}}
        i & 1 & 2 & 3 & 4 & 5 & 6 & 7 & 8 & 9 & 10 & 11 & 12 & 13 & 14 & 15 & 16 \\
        \hline
        u[i] &\texttt{~a~}&\texttt{~b~}&\texttt{~a~}&\texttt{~a~}&\texttt{~b~}&\texttt{~a~}&\texttt{~a~}&\texttt{~a~}&
        \texttt{~b~}&\texttt{~b~}&  \texttt{~a~}&\texttt{~a~}&\texttt{~b~}&\texttt{~a~}&\texttt{~a~}&\texttt{~b~}\\
        \Period[i]   & 1 & 2 & 2 & 3 & 3 & 3 & 3 & 7 & 7 & 10 & 10 & 11 & 11 & 11 & 11 & 11\\
        \Border[i]   & 0 & 0 & 1 & 1 & 2 & 3 & 4 & 1 & 2 &  0 &  1 &  1 &  2 &  3 &  4 &  5\\
        \Cover[i]    & 1 & 2 & 3 & 4 & 5 & 3 & 4 & 8 & 9 & 10 & 11 & 12 & 13 & 14 & 15 & 16\\
        \Cover^M[i]  & 0 & 0 & 0 & 0 & 0 & 3 & 4 & 0 & 0 & 0 & 0 & 0 & 0 & 0 & 0 & 0\\
        \LSeed[i]    & 1 & 2 & 2 & 3 & 3 & 3 & 3 & 4 & 4 & 10 & 10 & 11 & 11 & 11 & 11 & 11\\
        \LSeed^M[i]  & 0 & 0 & 2 & 3 & 4 & 5 & 6 & 7 & 8 &  0 & 10 & 11 & 12 & 13 & 14 & 15\\
        \Seed[i]     & 1 & 2 & 2 & 3 & 3 & 3 & 3 & 4 & 4 &  8 &  8 &  8 &  8 &  8 &  8 & 11\\
      \end{array}
    $\\
    \caption{
      \label{tab:tab1}
      An example string together with its periodic and quasiperiodic arrays.
      Note that the left-seed array and the seed array are non-decreasing.
    }
    \end{table}

    \noindent
    The border array, suffix border array and period array can be computed in $O(n)$ time
    \cite{AlgorithmsOnStrings,Jewels}.
    Apostolico and Breslauer \cite{DBLP:conf/birthday/ApostolicoB97,DBLP:journals/ipl/Breslauer92}
    gave an on-line $O(n)$ time algorithm computing the cover array $\Cover [1 \dotdot n]$ of a string.
    Li and Smyth~\cite{DBLP:journals/algorithmica/LiS02} provided an algorithm,
    having the same characteristics, for computing the longest cover array $\Cover^M [1 \dotdot n]$ of a given string.
    Note that the array $\Cover^M$ enables computing all covers of all prefixes of the string,
    same property holds for the border array $\Border$.
    Unfortunately, the $\LSeed^M$ array does not share this property.

    Table~\ref{tab:tab1} shows the above defined arrays for $u=\texttt{abaabaaabbaabaab}$.
    For example, for the prefix $u[1 \dotdot 13]$ the period equals 11, the border is \texttt{ab},
    the cover is \texttt{abaabaaabbaab}, the left seed is \texttt{abaabaaabba},
    the longest left seed is \texttt{abaabaaabbaa}, and the seed is \texttt{baabaaab}.

    We list here several useful (though obvious) properties of covers and seeds.

     \begin{observation}\label{obs:properties}~\\
       (a) A cover of a cover of $u$ is also a cover of $u$.\\
       (b) A cover of a left (right) seed of $u$ is also a left (right) seed of $u$.\\
       (c) A cover of a seed of $u$ is also a seed of $u$.\\
       (d) If $u$ is a factor of $v$ then $\seed(u)  \le \seed(v)$.\\
       (e) If $u$ is a prefix of $v$ then $\lseed(u) \le \lseed(v)$.\\
       (f) If $s$ and $s'$ are two covers of a string $u$, $|s'|<|s|$, then $s'$ is a cover of $s$.\\
       (g) If $s$ is the shortest cover or the shortest left seed or the shortest seed of a string $u$ then
       $\per(s) > |s|/2$.
      \end{observation}

    \noindent
    For a set $X$ of positive integers, let us define the \emph{maxgap} of $X$ as:
    $$\maxgap(X)\ =\ \max\{b-a : a, b \mbox{ are consecutive numbers in } X\} \ 
    \mbox{ or }0\mbox{ if }|X| \le 1.$$
    For example $\maxgap(\{1,3,8,13,17\})=5$.

    For a factor $v$ of $u$, let us define $\Occ(v,u)$ as the set of starting positions of all
    occurrences of $v$ in $u$.
    By $\firstocc(v)$ and $\lastocc(v)$ we denote $\min\Occ(v,u)$ and $\max \Occ(v,u)$ 
    respectively.
    For the sake of simplicity, we will abuse the notation, and denote
    $\maxgap(v) = \maxgap(\Occ(v,u))$.

\begin{figure}
  \begin{center}
    \includegraphics{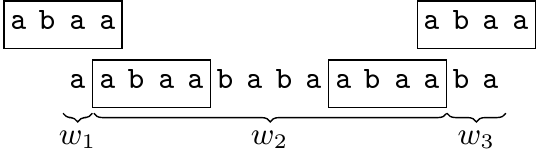}
  \end{center}
  \caption{\label{fig:border-seed}
    The word $s=\mathtt{abaa}$ is a border seed of $u=\mathtt{aabaababaabaaba}$.
  }
\end{figure}

    Assume $s$ is a factor of $u$.
    Let us decompose the word $u$ into $w_1w_2w_3$, where $w_2$ is the longest
    factor of $u$ for which $s$ is a border, i.e., 
    $w_2 = u[\firstocc(s) \dotdot (\lastocc(s)+|s|-1)]$.
    Then we say that $s$ is a \emph{border seed} of $u$ if $s$ is a seed of 
    $w_1\cdot s\cdot w_3$, see Fig.~\ref{fig:border-seed}.
    The following fact is a corollary of Lemma~\ref{lem:covers-and-lseeds}, 
    proved in Section \ref{sec:LSeed}.

    \begin{fact}\label{fact:border-seed}
      Let $s$ be a factor of $u \in \Sigma^*$.
      The word $s$ is a border seed of $u$ if and only if\ \ 
      $|s|\ \ge\ \max(\Period\,[\firstocc(s)+|s|-1],\ \Period'[\lastocc(s)]).$
    \end{fact}

    \smallskip
    \noindent
    Notions of maxgaps and border seeds provide a useful characterization of seeds.

    \begin{observation}\label{obs:maxgap-border-seed}
      Let $s$ be a factor of $u \in \Sigma^*$.
      The word $s$ is a seed of $u$ if and only if $|s| \ge \maxgap(s)$
      and $s$ is a border seed of $u$.
    \end{observation}

    \noindent
    Several new and efficient algorithms related to seeds in strings are presented in this paper.
    Linear time algorithms computing left-seed array and longest left-seed array are given
    in Section~\ref{sec:LSeed}.
    In Section~\ref{sec:SeedArray} we show a linear time algorithm finding seed-of-a-given-length
    and apply it to computing the seed array of a string in $O(n^2)$ time.
    Finally, in Section~\ref{sec:new-seeds} we describe an alternative simple $O(n\log{n})$ time computation
    of the shortest seed, from which we obtain an $O(n\log{(n/m)})$ time algorithm checking if
    the shortest seed has length at least $m$ (described in Section~\ref{sec:long-seeds}).

  \section{Computing Left-Seed Arrays} \label{sec:LSeed}

  In this section we show two $O(n)$ time algorithms for computing the left-seed
  array and an $O(n)$ time algorithm  for computing the longest left-seed array of a given string $u \in \Sigma^n$.
  We start by a simple characterization of the length of the shortest left
  seed of the whole string $u$ --- see Lemma \ref{lem:lseed}.
  In its proof we utilize the following auxiliary lemma which shows
  a correspondence between the shortest left seed of $u$ and shortest covers of all
  prefixes of $u$.

  \begin{lemma}\label{lem:covers-and-lseeds}
    Let $s$ be a prefix of $u$, and let $j$ be the length of the longest
    prefix of $u$ covered by $s$.
    Then $s$ is a left seed of $u$ if and only if $j \ge \per(u)$.

    In particular, the shortest left seed $s$ of $u$ is the shortest cover of
    the corresponding prefix $u[1 \dotdot j]$.
  \end{lemma}

  \begin{proof}
    ($\Rightarrow$)
    If $s$ is a left seed of $u$ then there exists a prefix $p$ of $s$ of length at least
    $n-j$ which is a suffix of $u$ (see Fig.~\ref{fig:lem-cbls-b}). 
    We use here the fact, that $u[1 \dotdot j]$ is the \emph{longest} prefix 
    of $u$ covered by $s$.
    Hence, $p$ is a border of $u$, and consequently
    $\border(u) \;\ge\; |p| \;\ge\; n-j$.
    Thus we obtain the desired inequality $j \ge \per(u)$.

\begin{figure}
\begin{center}
  \includegraphics{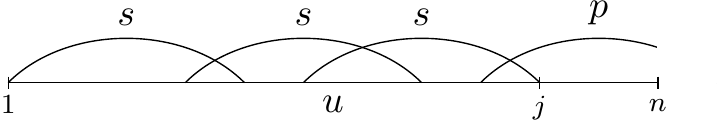}
\end{center}
\caption{\label{fig:lem-cbls-b}
  Illustration of part ($\Rightarrow$) of Lemma \ref{lem:covers-and-lseeds}.
}
\end{figure}

    ($\Leftarrow$)
    The inequality $j \ge \per(u)$ implies that $v=u[1 \dotdot j]$ is a left seed of $u$
    (see Fig.~\ref{fig:lem-cbls-a}).
    Hence, by Observation \ref{obs:properties}b, the word $s$, which is a cover of $v$,
    is also a left seed of~$u$.

    Finally, the ``in particular'' part is a consequence of Observation~\ref{obs:properties},
    parts b and f.
  \qed
  \end{proof}

\begin{figure}
\begin{center}
  \includegraphics{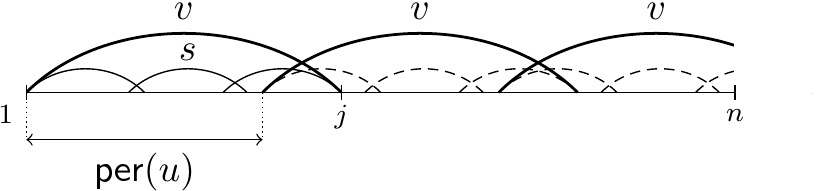}
\end{center}
\caption{\label{fig:lem-cbls-a}
  Illustration of part ($\Leftarrow$) of Lemma \ref{lem:covers-and-lseeds}.
}
\end{figure}

  \begin{lemma}\label{lem:lseed}
    Let $u \in \Sigma^n$ and let $\Cover[1 \dotdot n]$ be its cover array.
    Then:
    \begin{equation}\label{eq:lseed}
      \lseed(u) = \min\{\Cover[j]\ :\ j \ge \per(u)\}.
    \end{equation}
  \end{lemma}

  \begin{proof}
    By Lemma~\ref{lem:covers-and-lseeds}, the length of the shortest left seed of $u$
    can be found among the values $\Cover[\per(u)],\ldots,\Cover[n]$.
    And conversely, for each of the values $\Cover[j]$ for $\per(u) \le j \le n$, there exists
    a left seed of $u$ of length $\Cover[j]$.
    Thus $\lseed(u)$ equals the minimum of these values, which yields the formula~\eqref{eq:lseed}.
  \qed
  \end{proof}

  \noindent
  Clearly, the formula \eqref{eq:lseed} provides an $O(n)$ time algorithm
  for computing the shortest left seed of the whole string $u$.
  We show that, employing some algorithmic techniques, one can use this formula to compute
  shortest left seeds for all prefixes of $u$, i.e., computing the left-seed array of $u$,
  also in $O(n)$ time.

  \begin{theorem}
    For $u \in \Sigma^n$, its left-seed array can be computed in $O(n)$ time.
  \end{theorem}

  \begin{proof}
    Applying \eqref{eq:lseed} to all prefixes of $u$, we obtain:
    \begin{equation}\label{eq:lseed_array}
      \LSeed[i] = \min\{\Cover[j]\ :\ \Period[i] \le j \le i\}.
    \end{equation}
    Recall that both the period array $\Period[1 \dotdot n]$ and the cover array $\Cover[1 \dotdot n]$
    of $u$ can be computed in $O(n)$ time \cite{DBLP:conf/birthday/ApostolicoB97,DBLP:journals/ipl/Breslauer92,AlgorithmsOnStrings,Jewels}.

    The minimum in the formula \eqref{eq:lseed_array} could be computed by data structures for
    Range-Minimum-Queries \cite{DBLP:conf/escape/FischerH07,DBLP:journals/jda/Sadakane07},
    however in this particular case we can apply a much simpler algorithm.
    Note that $\Period[i-1] \le \Period[i]$,
    therefore the intervals of the form $[\Period[i],i]$ behave like a sliding window, i.e.,
    both their endpoints are non-decreasing.
    We use a bidirectional queue $Q$ which stores left-minimal elements in the current interval
    $[\Period[i],i]$ (w.r.t.\ the value $\Cover[j]$).
    In other words, elements of $Q$ are increasing and if $Q$ during the step $i$ contains an
    element $j$ then $j \in [\Period[i],i]$ and $\Cover[j] < \Cover[j']$ for all $j < j' \le i$.
    We obtain an $O(n)$ time algorithm \mbox{ComputeLeftSeedArray}.
    \qed
  \end{proof}

{\small
      \begin{center}
      \fbox{
      \begin{minipage}{10.5cm}
      \begin{algorithmic}[1]
      \vspace*{0.2cm}
        \REQUIRE ComputeLeftSeedArray($u$)
        \vspace*{0.2cm}
        \STATE $\Period[1 \dotdot n]:=$ period array of $u$;\ \ $\Cover[1 \dotdot n]:=$ cover array of $u$;
        \STATE $Q:=\mathit{emptyBidirectionalQueue}$;
        \FOR{$i:=1$ \textbf{to} $n$}
          \STATE \textbf{while} (\textbf{not} $\mathit{empty}(Q)$) \textbf{and} ($\mathit{front}(Q) < \Period[i]$)\ 
            \textbf{do} $\mathit{popFront}(Q)$;
          \STATE \textbf{while} (\textbf{not} $\mathit{empty}(Q)$) \textbf{and} ($\Cover[\mathit{back}(Q)] \ge \Cover[i]$)\ 
            \textbf{do} $\mathit{popBack}(Q)$;
          \STATE $\mathit{pushBack}(Q, i)$;
          \STATE $\LSeed[i]:=\Cover[\mathit{front}(Q)]$;
          \STATE \COMMENT{ $Q$ stores left-minimal elements of the interval $[\Period[i],i]$ }
        \ENDFOR
        \RETURN $\LSeed[1 \dotdot n]$;
      \vspace*{0.2cm}
      \end{algorithmic}
      \end{minipage}
      }
      \end{center}
}

  \medskip
  \noindent
  Now we proceed to an alternative algorithm computing the left-seed array, which also
  utilizes the criterion from Lemma~\ref{lem:covers-and-lseeds}.
  We start with an auxiliary algorithm \mbox{ComputeR-Array}.
  It computes an array $\R[1 \dotdot n]$ which stores, as $\R[i]$, the length of the longest prefix
  of $u$ for which $u[1 \dotdot i]$ is the shortest cover, $0$ if none.

{\small
      \begin{center}
      \fbox{
      \begin{minipage}{9cm}
      \begin{algorithmic}[1]
      \vspace*{0.2cm}
        \REQUIRE ComputeR-Array($u$)
        \vspace*{0.2cm}
        \STATE $\Cover[1 \dotdot n]:=$ cover array of $u$;
        \STATE \textbf{for} $i:=1$ \textbf{to} $n$ \textbf{do} $\R[i]:=0$;
        \STATE \textbf{for} $i:=1$ \textbf{to} $n$ \textbf{do} $\R[\Cover[i]]:=i$;
        \RETURN $\R[1 \dotdot n]$;
      \vspace*{0.2cm}
      \end{algorithmic}
      \end{minipage}
      }
      \end{center}
}

  \noindent
  The algorithm \mbox{Alternative-ComputeLeftSeedArray} computes the array $\LSeed$ from left to right.
  The current value of $\LSeed[i]$ is stored in the variable $\mathit{ls}$, note that this value
  never decreases (by Observation~\ref{obs:properties}e).
  Equivalently, for each $i$ we have $\LSeed[i-1] \le \LSeed[i] \le i$.

  The particular value of $\LSeed[i]$ is obtained using the necessary and sufficient condition from
  Lemma~\ref{lem:covers-and-lseeds}: $\LSeed[i]=\mathit{ls}$ if $\mathit{ls}$ is the smallest
  number such that $|w| \ge \per(u[1 \dotdot i]) = \Period[i]$, where $w$ is the longest prefix of
  $u[1 \dotdot i]$ that is covered by $u[1 \dotdot \mathit{ls}]$.
  We slightly modify this condition, substituting $w$ with the longest prefix $w'$ of the very word $u$
  that is covered by $u[1 \dotdot \mathit{ls}]$.
  Thus we obtain the condition $\R[\mathit{ls}] \ge \Period[i]$ utilized in the pseudocode below.

{\small
      \begin{center}
      \fbox{
      \begin{minipage}{10cm}
      \begin{algorithmic}[1]
      \vspace*{0.2cm}
        \REQUIRE Alternative-ComputeLeftSeedArray($u$)
        \vspace*{0.2cm}
        \STATE $\Period[1 \dotdot n]:=$ period array of $u$;\ \ $\R[1 \dotdot n]:=$ ComputeR-Array($u$);
        \STATE $\LSeed[0]:=0$;\ \ $\mathit{ls}:=0$;
        \FOR{$i:=1$ \textbf{to} $n$}
          \STATE \COMMENT{ An invariant of the loop: $\mathit{ls}=\LSeed[i-1]$. }
          \STATE \textbf{while} $\R[\mathit{ls}]<\Period[i]$ \textbf{do} $\mathit{ls}:=\mathit{ls}+1$;
          \STATE $\LSeed[i]:=\mathit{ls}$;
        \ENDFOR
        \RETURN $\LSeed[1 \dotdot n]$;
      \vspace*{0.2cm}
      \end{algorithmic}
      \end{minipage}
      }
      \end{center}
}

  \begin{theorem}
    Algorithm \mbox{Alternative-ComputeLeftSeedArray} runs in linear time.
  \end{theorem}

  \begin{proof}
    Recall that the arrays $\Period[1 \dotdot n]$ and $\Cover[1 \dotdot n]$ can be computed in linear time
    \cite{DBLP:conf/birthday/ApostolicoB97,DBLP:journals/ipl/Breslauer92,AlgorithmsOnStrings,Jewels}.
    The array $\R[1 \dotdot n]$ is obviously also computed in linear time.

    It suffices to prove that the total number of steps of the while-loop in the algorithm
    \mbox{Alternative-ComputeLeftSeedArray} is linear in terms of $n$.
    In each step of the loop, the value of $\mathit{ls}$ increases by one; this variable
    never decreases and it cannot exceed $n$.
    Hence, the while-loop performs at most $n$ steps and the whole algorithm runs in $O(n)$ time.
  \qed
  \end{proof}

  \medskip
  \noindent
  Concluding this section, we describe a linear-time algorithm computing the
  longest left-seed array, $\LSeed^M[1 \dotdot n]$, of the string $u \in \Sigma^n$.
  The following lemma gives a simple characterization of the length of the longest
  left seed of the whole string $u$.

  \begin{lemma}\label{lem:LSM}
    Let $u \in \Sigma^n$.
    If $\per(u)<n$ then $\lseedmax(u)=n-1$, otherwise $\lseedmax(u)=0$.
  \end{lemma}

  \begin{proof}
    First consider the case $\per(u)=n$.
    We show that $\lseed(u)=n$, consequently $\lseedmax(u)$ equals $0$.
    Assume to the contrary that $\lseed(u)<n$.
    Then, a non-empty prefix of the minimal left seed of $u$, say
    $w$, is a suffix of $u$ (consider the occurrence of the left seed that covers $u[n]$).
    Hence, $n - |w|$ is a period of $u$, a contradiction.

    Assume now that $\per(u)<n$.
    Then $u$ is a prefix of the word $u[1 \dotdot \per(u)]\cdot u[1 \dotdot n-1]$
    which is covered by $u[1 \dotdot n-1]$.
    Therefore $u[1 \dotdot n-1]$ is a left seed of $u$, $\lseedmax(u) \ge n-1$,
    consequently $\lseedmax(u)=n-1$.
    \qed
  \end{proof}

  \noindent
  Using Lemma~\ref{lem:LSM} we obtain $\LSeed^M[i] = i-1$ or $\LSeed^M[i] = 0$ for every $i$,
  depending on whether $\Period[i]<i$ or not.
  We obtain the following result.

  \begin{theorem}
    Longest left-seed array of $u \in \Sigma^n$ can be computed in $O(n)$ time.
  \end{theorem}

  \section{Computing Seeds of Given Length and Seed Array} \label{sec:SeedArray}
  In this section we show an $O(n^2)$ time algorithm computing the seed array $\Seed[1 \dotdot n]$
  of a given string $u \in \Sigma^n$, note that a trivial approach --- computing
  the shortest seed for every prefix of $u$ --- yields $O(n^2\log{n})$ time complexity.
  In our solution we utilize a subroutine: testing whether $u$ has a seed of a given length $k$.
  The following theorem shows that this test can be performed in $O(n)$ time.

  \begin{theorem}\label{thm:seed-of-a-given-length}
    It can be checked whether a given string $u \in \Sigma^n$ has a seed of a given length $k$
    in $O(n)$ time.
  \end{theorem}
  \begin{proof}
    Assume we have already computed in $O(n)$ time the suffix array $\SUF$ and the $\LCP$ array
    of longest common prefixes, see \cite{AlgorithmsOnStrings}.
    In the algorithm we start by dividing all factors of $u$ of length $k$ into groups
    corresponding to equal words.
    Every such group can be described as a maximal interval $[i \dotdot j]$ in the suffix array
    $\SUF$, such that each of the values $\LCP[i+1],\LCP[i+2],\ldots,\LCP[j]$ is at least $k$.
    The collection of such intervals can be constructed in $O(n)$ time by a single traversal
    of the $\LCP$ and $\SUF$ arrays (lines 1--9 of Algorithm SeedsOfAGivenLength).
    Moreover, using Bucket Sort, we can transform this representation
    into a collection of lists, each of which describes the set $\Occ(v,u)$ 
    for some factor $v$ of $u$, $v \in \Sigma^k$ (lines 10--11 of the algorithm).
    This can be done in linear time, provided that we use the same set of buckets in 
    each sorting and initialize them just once.

    Now we process each of the lists separately, checking the conditions from
    Observation~\ref{obs:maxgap-border-seed}: in lines 14--18 of the algorithm we check the
    ``maxgap'' condition, and in line 19 the ``border seed'' condition, employing
    Fact~\ref{fact:border-seed}.

    Thus, having computed the arrays $\SUF$ and $\LCP$, and the period arrays 
    $\Period[1 \dotdot n]$ and $\Period'[1 \dotdot n]$ of $u$, 
    we can find all seeds of $u$ of length $k$ in $O(n)$ total time.
  \qed
  \end{proof}

  {\small
    \begin{center}
      \fbox{
      \begin{minipage}{11cm}
      \begin{algorithmic}[1]
      \vspace*{0.2cm}
        \REQUIRE SeedsOfAGivenLength($u$, $k$)
        \vspace*{0.2cm}
        \STATE $\Period[1 \dotdot n]:=$ period array of $u$;\ \ $\Period'[1 \dotdot n]:=$ suffix period array of $u$;
        \STATE $\SUF[1 \dotdot n]:=$ suffix array of $u$;\ \ $\LCP[1 \dotdot n]:=$ lcp array of $u$;
        \STATE $\mathit{Lists}:=\mathit{emptyList}$;
        \STATE $j:=1$;
        \WHILE{$j \le n$}
          \STATE $\mathit{List} := \{\SUF[j]\}$;
          \WHILE{$j<n$ \textbf{and} $\LCP[j+1] \ge k$}
            \STATE $j:=j+1$;\ \ $\mathit{List}:=\mathit{append}(\mathit{List},\SUF[j])$;
          \ENDWHILE
          \STATE $j:=j+1$;\ \ $\mathit{Lists}:=\mathit{append}(\mathit{Lists},\mathit{List})$;
        \ENDWHILE
        \FORALL{$\mathit{List}\ \textbf{in}\ \mathit{Lists}$}
          \STATE BucketSort($\mathit{List}$);
          \COMMENT{ using the same set of buckets }
        \ENDFOR
        \FORALL{$\mathit{List}\ \textbf{in}\ \mathit{Lists}$}
          \STATE $\firstocc:=\mathit{prev}:=n$;\ \ $\lastocc:=1;$\ \ $\mathit{covers}:=\mathbf{true}$;
          \FORALL{$i\ \textbf{in}\ \mathit{List}$}
            \STATE $\firstocc:=\min(\firstocc,i)$;\ \ $\lastocc:=\max(\lastocc,i)$;
            \IF{$i>\mathit{prev}+k$}
              \STATE $\mathit{covers}:=\mathbf{false}$;
            \ENDIF
            \STATE $\mathit{prev}:=i$;
          \ENDFOR
          \IF{$\mathit{covers}$ \textbf{and} ($k \ge \max(\Period[\firstocc+k-1],\ \Period'[\lastocc])$)}
            \PRINT ``$u[\firstocc \dotdot (\firstocc+k-1)]$ is a seed of $u$'';
        \vspace*{0.2cm}
          \ENDIF
        \ENDFOR
      \vspace*{0.2cm}
      \end{algorithmic}
      \end{minipage}
      }
      \end{center}
}

  \noindent
  We compute the elements of the seed array $\Seed[1 \dotdot n]$ from left to right, i.e.,
  in the order of increasing lengths of prefixes of $u$.
  Note that $\Seed[i+1] \ge \Seed[i]$ for any $1 \le i \le n-1$, this is due to
  Observation~\ref{obs:properties}d.
  If $\Seed[i+1] > \Seed[i]$ then we increase the current length of the seed
  by one letter at a time, in total at most $n-1$ such operations are performed.
  Each time we query for existence of a seed of a given length using the algorithm
  from Theorem~\ref{thm:seed-of-a-given-length}.
  Thus we obtain $O(n^2)$ time complexity.

  \begin{theorem}
    The seed array of a string $u \in \Sigma^n$ can be computed in $O(n^2)$ time.
  \end{theorem}

  \section{Alternative Algorithm for Shortest Seeds}\label{sec:new-seeds}
    In this section we present a new approach to shortest seeds computation
    based on very simple independent processing of disjoint chains in the suffix tree.
    It simplifies the computation of shortest seeds considerably.

    Our algorithm is also based on a slightly modified version of Observation~\ref{obs:maxgap-border-seed},
    formulated below as Lemma~\ref{lem:prefmaxgap-border-seed}, which allows to relax the definition of maxgaps.
    We discuss an algorithmically easier version of maxgaps, called prefix maxgaps, 
    and show that it can substitute $\maxgap$ values when looking for the shortest seed. 
    
    We start by analyzing the ``border seed'' condition.
    We introduce somewhat more abstract representation of sets of factors of $u$, 
    called \emph{prefix families}, and show how to find in them the shortest border seeds of $u$. 
    Afterwards the key algorithm for computing prefix maxgaps is presented.
    Finally, both techniques are utilized to compute the shortest seed. 

    Let us fix the input string $u \in \Sigma^n$.
    For $v \in \Sigma^*$, by $\PREF(v)$ we denote the set of all prefixes of $v$
    and by $\PREF(v,k)$ we denote $\PREF(v) \cap \Sigma^k\Sigma^*$ (\emph{limited prefix subset}).

    Let $\FF$ be a family of limited prefix subsets of some factors of $u$, we call $\FF$ a \emph{prefix family}.
    Every element $\PREF(v,k) \in \FF$ can be represented in a canonical form, by
    a tuple of integers: $(\firstocc(v), \lastocc(v), k, |v|)$.
    Such a representation requires only constant space per element.
    By $\bseed(u,\FF)$ we denote the shortest border seed of $u$ contained in some 
    element of $\FF$.

    \begin{example}
      Let $u = \mathtt{aabaababaabaaba}$ be the example word from Fig.~\ref{fig:border-seed}.
      Let:
      $$
        \FF = 
        \{\PREF(\mathtt{abaab},4),\ \PREF(\mathtt{babaa},4)\} =
        \{(2,10,4,5),\ (6,6,4,5)\}.
      $$
      Note that $\bigcup\FF = \{\mathtt{abaa},\mathtt{abaab},\mathtt{baba},\mathtt{babaa}\}$.
      Then $\bseed(u,\FF) = \mathtt{abaa}$.
    \end{example}

    \noindent
    The proof of the following fact is present implicitly in \cite{DBLP:journals/algorithmica/IliopoulosMP96}
    (type-A and type-B seeds).

    \begin{theorem}\label{thm:min-border-seeds}
      Let $u \in \Sigma^n$ and let $\FF$ be a prefix family given in a canonical form.
      Then $\bseed(u,\FF)$ can be computed in linear time.
    \end{theorem}

\paragraph{{\bf Alternative proof of Theorem~\ref{thm:min-border-seeds}.} }
  There is an alternative algorithm for computing $\bseed(u,\FF)$, 
  based on a special version of Find-Union data structure.
  Recall that $\Border[1 \dotdot n]$ is the border-array of $u$.
  Denote by $\Q({\cal I},c)$ (\emph{first-greater-equal}) a query:
  $$\Q({\cal I},c)\ =\ \min\{ i\; :\; i\in {\cal I},\; \Border[i]\ge c\},$$
  where ${\cal I}$ is a subinterval of $[1 \dotdot n]$.
  We assume that $\min\emptyset = +\infty$.
  A sequence of linear number of such queries, sorted according to non-decreasing
  values of $c$, can be easily answered in linear time, using an interval version of
  Find-Union data structure, see \cite{DBLP:conf/spire/CrochemoreIKRRW10,GT:83}.
  The following algorithm applies the condition for border seed from Fact~\ref{fact:border-seed}
  to every element of $\FF$, with $\Period[\firstocc(s)+|s|-1]$ substituted by 
  $\firstocc(s)+|s|-1-\Border[\firstocc(s)+|s|-1]$.
  We omit the details. \qed

  {\small
      \begin{center}
      \fbox{
      \begin{minipage}{11cm}
      \begin{algorithmic}[1]
      \vspace*{0.2cm}
        \REQUIRE ComputeBorderSeed($u$, $\FF$)
        \vspace*{0.2cm}
        \STATE $\mathit{bseed}:=+\infty$;
        \FORALL{$(\firstocc(v), \lastocc(v), k, |v|) \mathbf{~in~} \FF$, 
          in non-decreasing order of $\firstocc(v)$}
          \STATE $k:=\max(\Period'[\lastocc(v)],\;k)$;
          \STATE $\mathcal{I}:= [\firstocc(v)+k-1,\; \firstocc(v)+|v|-1]$;
          \STATE $\mathit{pos} := \Q(\mathcal{I},\ \firstocc(v)-1)$;
          \STATE $\mathit{bseed}:= \min (\mathit{bseed},\ \mathit{pos}-\firstocc(v)+1)$;
        \ENDFOR
        \RETURN $\mathit{bseed}$;
      \vspace*{0.2cm}
      \end{algorithmic}
      \end{minipage}
      }
      \end{center}
}

\medskip
\paragraph{\bf Computation of the shortest seeds via prefix maxgaps.}
  \noindent
  Let $\Tree(u)$ be the suffix tree of $u$, recall that it can
  be constructed in $O(n)$ time \cite{AlgorithmsOnStrings,Jewels}.
  By $\Nodes(u)$ we denote the set of factors of $u$ corresponding to explicit nodes
  of $\Tree(u)$, for simplicity we identify the nodes with the strings they represent.
  For $v \in \Nodes(u)$, the set $\Occ(v,u)$ corresponds to leaf list of the node $v$
  (i.e., the set of values of leaves in the subtree rooted at $v$), denoted as $\LL(v)$.
  Note that $\firstocc(v)=\min\LL(v)$ and $\lastocc(v)=\max\LL(v)$, and
  such values can be computed for all $v \in \Nodes(u)$ in $O(n)$ time.
  For $v \in \Nodes(u)$, we define the \emph{prefix maxgap} of $v$ as:
  $$
    \prefmaxgap(v)\ =\ \max\{\maxgap(w)\;:\; w \in \PREF(v)\}.
  $$
  Equivalently, $\prefmaxgap(v)$ is the maximum of $\maxgap$ values on the path from $v$
  to the root of $\Tree(u)$.
  We introduce an auxiliary problem:

  \paragraph{\bf Prefix Maxgap Problem:}~\\
  \indent
  given a word $u \in \Sigma^n$,
  compute $\prefmaxgap(v)$ for all $v \in \Nodes(u)$.

\vskip 0.3cm \noindent
  The following lemma (an alternative formulation of Observation~\ref{obs:maxgap-border-seed}) shows that
  prefix maxgaps can be used instead of maxgaps in searching for seeds.
  This is important since computation of prefix maxgaps $\prefmaxgap(v)$ is simple, 
  in comparison with $\maxgap(v)$ --- this is due to the fact that the $\prefmaxgap(v)$ values
  on each path down the suffix tree $\Tree(u)$ are non-decreasing.
  Efficient computation of $\maxgap(v)$ requires using augmented height-balanced trees \cite{DBLP:conf/cpm/BrodalP00}
  or other rather sophisticated techniques \cite{DBLP:journals/iandc/BerkmanIP95}.
  The shortest-seed algorithm in \cite{DBLP:journals/algorithmica/IliopoulosMP96} also
  computes prefix maxgaps instead of maxgaps, however this observation is missing in
  \cite{DBLP:journals/algorithmica/IliopoulosMP96}.

  \begin{lemma}\label{lem:prefmaxgap-border-seed}
    Let $s$ be a factor of $u \in \Sigma^*$ and let $w$ be the shortest element
    of $\Nodes(u)$ such that $s \in \PREF(w)$.
    The word $s$ is a seed of $u$ if and only if $|s| \ge \prefmaxgap(w)$
    and $s$ is a border seed of $u$.
  \end{lemma}
  \begin{proof}
    If $s$ corresponds to an element of $\Nodes(u)$, then $s=w$. 
    Otherwise, $s$ corresponds to an implicit node in an edge in the suffix tree, and 
    $w$ is the lower end of the edge. 
    Note that in both cases we have $\prefmaxgap(w) \ge \maxgap(w) = \maxgap(s)$. 
    By Observation~\ref{obs:maxgap-border-seed},
    this implies part ($\Leftarrow$) of the conclusion.
    As for the part ($\Rightarrow$), it suffices to show that $|s| \ge \prefmaxgap(w)$.

    Assume, to the contrary, that $|s| < \prefmaxgap(w)$.
    Let $v \in \PREF(w) \cap \Nodes(u)$ be the word for which $\maxgap(v)=\prefmaxgap(w)$,
    and let $a,b$ be consecutive elements of the set $\Occ(v,u)$ for which $a+\maxgap(v)=b$.

    Let us note that no occurrence of $s$ starts at any of the positions $a+1,\ldots,b-1$.
    Moreover, none of the suffixes of the form $u[i \dotdot n]$, for $a+1 \le i \le b-1$,
    is a prefix of $s$. 
    Indeed, $v$ is a prefix of $s$ of length at most $n-b+1$, 
    and such an occurrence of $s$ (or its prefix) would imply an extra occurrence of $v$. 
    Note that at most $|s| \le b-a-1$ first positions in the interval
    $[a,b]$ can be covered by an occurrence of $s$ in $u$ (at position $a$ or earlier) or by
    a suffix of $s$ which is a prefix of $u$.
    Hence, position $b-1$ is not covered by $s$ at all, a contradiction.
    \qed
  \end{proof}

  \noindent
  By Lemma~\ref{lem:prefmaxgap-border-seed}, to complete the shortest seed algorithm
  it suffices to solve the Prefix Maxgap Problem (this is further clarified in the
  ComputeShortestSeed algorithm below).
  For this, we consider the following problem.
  By $\SORT(X)$ we denote the sorted sequence of elements of $X\subseteq \{1,2,\ldots,n\}$.
\vskip 0.3cm
  \smallskip
  \noindent
  {\bf Chain Prefix Maxgap Problem}

  {\bf Input:}\ a family of disjoint sets $X_1, X_2,\ldots, X_k\subseteq \{1,2,\ldots,n\}$ \\
  \hspace*{1.6cm} together with $\SORT(X_1 \cup X_2 \cup \ldots \cup X_k)$.\\
  \hspace*{1.6cm} The size of the input is $m = \sum|X_i|$.

  {\bf Output:}\ the numbers $\prefmaxgap_i\ =\ \max_{j\le i}\; \maxgap(X_j\cup X_{j+1}\cup\ldots\cup X_k)$.
\vskip 0.3cm
  \begin{theorem}\label{thm:cpmp}
    The Chain Prefix Maxgap Problem can be solved in $O(m)$ time
    using an auxiliary array of size $n$.
  \end{theorem}
  \begin{proof}
    Initially we have the list $L = \SORT(X_1 \cup X_2 \cup \ldots \cup X_k)$.
    Let $\mathit{pred}$ and $\mathit{suc}$ denote the predecessor and successor of an element of $L$.
    The elements of $L$ store a Boolean flag $\mathit{marked}$, initially set to false.
    In the algorithm we use an auxiliary array $\mathit{pos}[1 \dotdot n]$ such that
    $\mathit{pos}[i]$ is a pointer to the element of value $i$ in $L$, if there is
    no such element then the value of $\mathit{pos}[i]$ can be arbitrary.
    Obviously the algorithm takes $O(m)$ time.
    \qed
  \end{proof}

    {\small
      \begin{center}
      \fbox{
      \begin{minipage}{8.5cm}
      \begin{algorithmic}[1]
      \vspace*{0.2cm}
        \REQUIRE ChainPrefixMaxgap($L$)
        \vspace*{0.2cm}
        \STATE $\prefmaxgap_1:=\maxgap(L)$; \COMMENT{ naive computation }
        \FOR{$j:=2$ \textbf{to} $k$}
          \STATE $\prefmaxgap_j:=\prefmaxgap_{j-1}$;
          \STATE \textbf{for all} $i$ \textbf{in} $X_{j-1}$ \textbf{do} $\mathit{marked}(\mathit{pos}[i]):=\mathbf{true}$;
          \FORALL{$i$ \textbf{in} $X_{j-1}$}
            \STATE $p:=\mathit{pred}(\mathit{pos}[i])$;\ \ $q:=\mathit{suc}(\mathit{pos}[i])$;
            \IF{($p \ne \textbf{nil}$) \textbf{and} ($q \ne \textbf{nil}$) \textbf{and}
                         (\textbf{not} $\mathit{marked}(p)$) \\ \textbf{and} (\textbf{not} $\mathit{marked}(q)$)}
              \STATE $\prefmaxgap_j:=\max(\prefmaxgap_j,\ \mathit{value}(q)-\mathit{value}(p))$;
            \ENDIF
            \STATE $\mathit{delete}(L,\mathit{pos}[i])$;
          \ENDFOR
        \ENDFOR
      \vspace*{0.2cm}
      \end{algorithmic}
      \end{minipage}
      }
      \end{center}
    }

 \begin{theorem}\label{thm:pmp}
   The Prefix Maxgap Problem can be reduced to a collection of Chain Prefix Maxgap Problems of total size $O(n \log n)$.
 \end{theorem}
 \begin{proof}
   We solve a more abstract version of the Prefix Maxgap Problem.
   We are given an arbitrary tree $T$ with $n$ leaves annotated with distinct
   integers from the interval $[1,n]$, and we need to compute the values $\prefmaxgap(v)$
   for all $v \in \Nodes(T)$, defined as follows:
   $\maxgap(v) = \maxgap(\LL(v))$, where $\LL(v)$ is the leaf list of $v$, and $\prefmaxgap(v)$
   is the maximum of the values $\maxgap$ on the path from $v$ to the root of $T$.
   We start by sorting $\LL(\mathit{root}(T))$, which can be done in $O(n)$ time.
   Throughout the algorithm we store a global auxiliary array $\mathit{pos}[1 \dotdot n]$,
   required in the ChainPrefixMaxgap algorithm.
   
   Let us find a \emph{heaviest path} $P$ in $T$, i.e., a path from the root down to a leaf, 
   such that all \emph{hanging} subtrees are of size at most $|T|/2$ each.
   The values of $\prefmaxgap(v)$ for $v\in P$ can all be computed in $O(n)$ time, using a reduction
   to the Chain Prefix Maxgap Problem (see Fig.~\ref{fig:cpmp_pmp}).
   
\begin{figure}
\begin{center}
\includegraphics{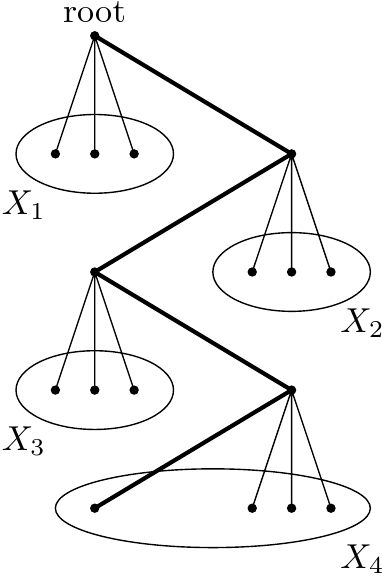}
\end{center}
\caption{\label{fig:cpmp_pmp}
  A tree with an example heaviest path $P$ (in bold).
  The values $\prefmaxgap(v)$ for $v \in P$ can be computed using a reduction
  to the Chain Prefix Maxgap Problem with the sets $X_1$ through $X_4$.
}
\end{figure}
   
   Then we perform the computation recursively for the hanging subtrees, 
   previously sorting $\LL(T')$ for each hanging subtree $T'$.
   Such sorting operations can be performed in $O(n)$ total time for all hanging subtrees.

   At each level of recursion we need a linear amount of time, 
   and the depth of recursion is logarithmic.
   Hence, the total size of invoked Chain Prefix Maxgap Problems is $O(n\log n)$.
   \qed
 \end{proof}
    
    \noindent
    Now we proceed to the shortest seed computation.
    In the algorithm we consider all factors of $u$, dividing them into groups corresponding
    to elements of $\Nodes(u)$.
    Let $w \in \Nodes(u)$ and let $v$ be its parent.
    Let $s \in \PREF(w)$ be a word containing $v$ as a proper prefix, i.e., $s \in \PREF(w,|v|+1)$.
    By Lemma~\ref{lem:prefmaxgap-border-seed}, the word $s$ is a seed of $u$ if and only if
    $|s| \ge \prefmaxgap(w)$ and $s$ is a border seed of $u$.
   
    Using the previously described reductions (Theorems~\ref{thm:min-border-seeds}--\ref{thm:pmp}),
    we obtain the following algorithm:

{\small
      \begin{center}
      \fbox{
      \begin{minipage}{11cm}
      \begin{algorithmic}[1]
      \vspace*{0.2cm}
        \REQUIRE ComputeShortestSeed($u$)
        \vspace*{0.2cm}
        \STATE Construct the suffix tree $T(u)$ for the input string $u$;
        \STATE Solve the Prefix Maxgap Problem for $T(u)$ using the ChainPrefixMaxgap 
        \STATE \hspace*{0.4cm} algorithm --- in $O(n\log n)$ total time (Theorems~\ref{thm:cpmp} and \ref{thm:pmp});
        \STATE $\FF:=\{\ \PREF(w,\; \max(|v|+1, \prefmaxgap(w)))\ :\ (v,w)\mbox{ is an edge in }\Tree(u)\ \};$
        \RETURN $\bseed(u,\FF)$;\ \ \COMMENT{ Theorem~\ref{thm:min-border-seeds} }
      \vspace*{0.2cm}
      \end{algorithmic}
      \end{minipage}
      }
      \end{center}
}

  \noindent
  Observe that the {\em workhorse} of the algorithm is the chain version
  of the Prefix Maxgap Problem, which has a fairly simple linear time solution.
  The main problem is of a structural nature, we have a collection of very simple
  problems each computable in linear time but the total size is not linear.
  This identifies the bottleneck of the algorithm from the complexity point of view.

  \section{Long Seeds}\label{sec:long-seeds}
  Note that the most time-expensive part of the ComputeShortestSeed algorithm is the computation of
  prefix maxgaps, all the remaining operations are performed in $O(n)$ time.
  Using this observation we can show a more efficient algorithm computing the shortest seed
  provided that its length $m$ is sufficiently large.
  For example if $m = \Theta(n)$ then we obtain an $O(n)$ time algorithm for the shortest seed.
  
  \begin{theorem}
    One can check if the shortest seed of a given string $u$ has length at least $m$
    in $O(n\log{(n/m)})$ time, where $n=|u|$.
    If so, a corresponding seed can be reported within the same time complexity.
  \end{theorem}
  \begin{proof}
    We show how to modify the ComputeShortestSeed algorithm.
    Denote by $s$ the shortest seed of $u$, $|s|=m$.

    By Observation~\ref{obs:properties}g, the longest overlap between consecutive
    occurrences of $s$ in $u$ is at most $\frac{m}2$, therefore the number of occurrences of
    $s$ in $u$ is at most $\frac{2n}{m}$.
    Hence, searching for the shortest seed of length at least $m$, it suffices to consider
    nodes $v$ of the suffix tree $\Tree(u)$ for which: $|v| \ge m$ and $|\LL(v)| \le \frac{2n}{m}$.
    
    Thus, we are only interested in prefix maxgaps for nodes in several subtrees of
    $\Tree(u)$, each of which contains $O(n/m)$ nodes.
    Thanks to the small size of each subtree, the algorithm ComputeShortestSeed
    finds all such prefix maxgaps in $O(n\log{(n/m)})$ time.
    Please note that using this algorithm for each node we obtain a prefix maxgap
    only in its subtree (not necessarily in the whole tree),
    however Lemma~\ref{lem:prefmaxgap-border-seed} can be simply adjusted to such
    a modified definition of prefix maxgaps.
    \qed
  \end{proof}

  \bibliographystyle{abbrv}
  \bibliography{seeds}

\end{document}